\documentclass[a4paper]{article}

\usepackage{fullpage}
\usepackage{amsmath,amsthm,paralist,bm,yhmath}
\usepackage{amssymb}
\usepackage{epstopdf}
\usepackage{epsfig}
\usepackage{graphicx}
\usepackage{cite}
\usepackage{empheq}
\usepackage[footnotesize]{caption}

\makeatletter
\let\NAT@parse\undefined
\makeatother
\usepackage[sort&compress, numbers]{natbib}

\usepackage[tight,footnotesize]{subfigure}
\usepackage{upgreek}
\usepackage[usenames]{color}
\newcommand{\Rbb}{\mathbb{R}}

\newcommand{\scp}[2]{\langle #1, #2 \rangle}

\newtheorem{theorem}{Theorem}

\newtheorem{proposition}{Proposition}
\newtheorem{corollary}{Corollary}

\newcommand{\supp}{{\rm supp}\,}
\newcommand{\tinv}[1]{{\textstyle\frac{1}{#1}}}
\newcommand{\sign}{{\rm sign}\,}

\newcommand{\ud}{\mathrm{d}} 

\newcommand{\ie}{\emph{i.e.}, } 
\newcommand{\eg}{\emph{e.g.}, }

\DeclareMathOperator*{\argmax}{argmax}

\DeclareMathOperator{\Id}{\bf Id}

\newcommand{\cl}{\mathcal}
\newcommand{\fk}{\mathfrak}
\newcommand{\bs}{\boldsymbol}
\newcommand{\bb}{\mathbb}

\newcommand{\sq}{\vspace{0mm}}
\newcommand{\sqe}{\vspace{0mm}}
\newcommand{\qpt}{q}
\newcommand{\xin}{\bs x_0}

\title{\vspace{0mm}Quantized Iterative Hard Thresholding:\\[1mm]
\LARGE Bridging 1-bit and High-Resolution\\[1mm] Quantized Compressed Sensing}
\author{Laurent Jacques$^*$, 
K\'evin Degraux and
Christophe De Vleeschouwer$^*$\\[1mm]
\small ICTEAM Institute, ELEN Department, Universit\'e catholique de
    Louvain (UCL)\\[-2mm]}
\date{\today}

\begin{document}

\maketitle

\begin{abstract}
In this work, we show that reconstructing a sparse signal from
quantized compressive measurement can be achieved in an unified
formalism whatever the (scalar) quantization resolution, \ie from 1-bit to high
resolution assumption. This is achieved by generalizing
the \emph{iterative hard thresholding} (IHT) algorithm and its binary variant
(BIHT) introduced in previous works to enforce the consistency of the
reconstructed signal with respect to the quantization model. The
performance of this algorithm, simply called \emph{quantized} IHT (QIHT), is evaluated in comparison with other
approaches (\eg IHT, \emph{basis pursuit denoise}) for several quantization scenarios.  
\end{abstract}

{\renewcommand{\thefootnote}{*}
\footnotetext{LJ and CDV are funded by the Belgian F.R.S-FNRS. Part of this research
is supported by the DETROIT project (WIST3), Walloon Region,
Belgium. {\em Acknowledgements:} We thank Prasad Sudhakar (UCL/ICTEAM)
and the anonymous reviewers of \textsc{Sampta} 2013 for their useful
comments. {\em Note:} This document is a preprint related to another work accepted in
Sampta13, Bremen, Germany.}}

\sq
\section{Introduction}
\label{sec:introduction}
\sq

Since the advent of Compressed Sensing (CS) almost 10 years ago~\cite{donoho2006cs,candes2008rip}, many
works have treated the problem of inserting this theory into an appropriate quantization scheme. This step is indeed
mandatory for transmitting, storing and even processing any
compressively acquired information, and more generally for  sustaining
the embedding of the CS principle in sensor design.   

In its most popular version, CS provides uniform theoretical guarantees for stably recovering
any sparse (or compressible) signal at a sensing rate
proportional to the signal intrinsic dimension (\ie its \emph{sparsity} level)
\cite{donoho2006cs,candes2008rip}. In this context, scalar
quantization of compressive measurements has been considered
along two main directions.

First, under a high-resolution quantization assumption, \ie when the
number of bits allocated to encode each measurement is high, the
quantization impact is often modeled as a mere additive Gaussian noise
whose variance is adjusted to the quantization $\ell_2$-distortion
\cite{Dai2009}. In short, under this high-rate model, the
CS stability guarantees under additive Gaussian noise, \ie as derived
from the $\ell_2-\ell_1$ instance optimality~\cite{candes2008rip}, are used to
bound the reconstruction error obtained from quantized
observations. Variants of these works handle quantization saturation
\cite{LasBouDav::2009::Demcracy-in-action}, prequantization noise
\cite{Zymnis2009}, $\ell_p$-distortion models ($p\geq 2$) for improved reconstruction in oversampled regimes
\cite{Jacques2010,jacques_arxiv_2012_nonunfibpdq}, optimize the
high-resolution quantization procedure~\cite
{SunGoy::20090::Optimal-quantization} or integrate more evolved
$\Sigma\Delta$-quantization models departing from scalar PCM quantization~\cite{gunturk2010sobolev}. 

Second, and more recently, extreme 1-bit quantization recording only the sign of the
compressive measurement, \ie an information encoded in a single bit,
has been considered~\cite{BouBar::2008::1-Bit-compressive, jacques2011robust, Plan2011, plan2012robust}. New guarantees have been
developed to tackle the non-linear nature of the sign operation thanks
to the replacement of the \emph{restricted isometric property} (RIP) by the
quasi-isometric \emph{binary $\epsilon$-stable embedding} (B$\epsilon$SE)~\cite{jacques2011robust}, or to
more general characterization of the binary embedding of sets based on their Gaussian Mean Width~\cite{Plan2011, plan2012robust}.
In this context, iterative methods such as the \emph{binary iterative hard
thresholding}~\cite{jacques2011robust} or linear programming optimization
\cite{Plan2011} have been introduced for estimating the 1-bit sensed signal.
  
This work proposes a general procedure for handling the reconstruction
of sparse signals observed according to a standard non-uniform scalar
quantization of the compressive measurements. The novelty of this
scheme is its ability to handle any resolution level, from 1-bit to
high-resolution, in a progressive fashion. Conversely to the Bayesian
approach of~\cite{yang2012unified}, our method relies on a
generalization of the \emph{iterative hard thresholding} (IHT)
\cite{blumensath2009iterative} that we simply called \emph{quantized iterative
hard thresholding}. Actually, QIHT reduces to BIHT for 1-bit sensing
and it converges to IHT at high resolution. 

\paragraph*{Conventions} Most of domain dimensions (\eg $M$, $N$) are
denoted by capital roman letters. Vectors and matrices are associated
to bold symbols while lowercase light letters are associated to scalar
values.  The $i^{\rm th}$ component of a vector $\bs u$ is $u_i$ or
$(\bs u)_i$. The identity matrix is $\Id$. 
The set of indices in $\Rbb^D$ is $[D]=\{1,\,\cdots,D\}$. Scalar product between
two vectors $\bs u,\bs v \in \Rbb^{D}$ reads $\bs u^* \bs v = \scp{\bs
  u}{\bs v}$ (using the transposition $(\cdot)^*$), while the Hadamard
product $\bs u \odot \bs v$ is such that $(\bs u \odot \bs v)_i = u_i v_i$. For any $p\geq 1$,
$\|\cdot\|_p$ represents the $\ell_p$-norm such that $\|\bs u\|_p^p =
\sum_i |u_i|^p$ with $\|\bs u\|=\|\bs u\|_2$ and $\|\bs u\|_\infty =
\max_i |u_i|$. The $\ell_0$ ``norm'' is $\|\bs u\|_0 = \# \supp \bs
u$, where $\#$ is the cardinality operator and $\supp \bs u = \{i: u_i
\neq 0\} \subseteq [D]$.  For $\cl S \subseteq [D]$, $\bs u_{\cl S}\in
\Rbb^{\#\cl S}$ (or $\bs \Phi_{\cl S}$) denotes the vector (resp. the
matrix) obtained by retaining the components (resp. columns) of $\bs
u\in\Rbb^D$ (resp. $\bs \Phi\in\Rbb^{D'\times D}$) belonging to $\cl
S\subseteq [D]$. The operator $\cl H_K$ is the hard thresholding operator setting all the
coefficients of a vector to 0 but those having the $K$ strongest
amplitudes. The set of canonical $K$-sparse vectors in $\Rbb^N$ is $\Sigma_K=\{\bs v\in\Rbb^N: \|\bs v\|_0 \leq
K\}$ while $\Sigma_{\cl T}$ denotes the set of vectors whose support
is $\cl T\subseteq [N]$. Moreover, $\Sigma^*_K=\Sigma_K\cap S^{N-1}$
and $\Sigma^*_{\cl T}=\Sigma^*_{\cl T}\cap S^{N-1}$ with $S^{N-1}$ the
$(N-1)$-sphere in $\Rbb^N$. Finally, $\chi_{\cl I}$ is the
characteristic function on $\cl I\subset \Rbb$, $\sign \lambda$ equals $1$ if
$\lambda$ is positive and $-1$ otherwise, $(\lambda)_+ = (\lambda +
|\lambda|)/2$ and $(\lambda)_- = -(-\lambda)_+$ project $\lambda$ on
$\Rbb_+$ and $\Rbb_-$, respectively, with all these operators being applied component wise onto vectors. 

\sq
\section{Noisy Compressed Sensing Framework}
\label{sec:noisy-compr-sens}
\sq

The \emph{iterative hard thresholding} (IHT) algorithm has been introduced for
iteratively reconstructing a sparse or compressible signal $\bs x\in\Rbb^N$ from
compressible observations $\bs y = \bs \Phi \bs x + \bs n$, where $\bs
\Phi\in\Rbb^{M \times N}$ is the sensing matrix and $\bs
n\in\Rbb^M$ stands for a possible observational noise with bounded
energy $\|\bs n\|\leq \varepsilon$. IHT is an
alternative to the \emph{basis pursuit denoise} (BPDN) method~\cite{Chen98atomic}
which aims at solving a global convex minimization promoting a $\ell_1$-sparse
data prior model under the constraint of reproducing the
compressive observation. 

Assuming that $\bs x$ is $K$-sparse in the canonical basis $\bs \Psi =
\Id$, \ie $\bs x \in \Sigma_K$, the IHT algorithm is designed to approximately solve the (LASSO-type) problem\sqe
\begin{equation}
  \label{eq:lasso-iht}
\min_{\bs u\in\Rbb^N} \tinv{2}\|\bs y - \bs \Phi\bs u\|^2\
\text{s.t.}\ \|\bs u\|_0\leq K.\sqe  
\end{equation}
It proceeds by computing the following recursion\sqe
\begin{align*}
\bs x^{(n+1)} = \cl H_K\big[\bs x^{(n)} + \mu\bs\Phi^*(\bs y - \bs\Phi\bs
x^{(n)})\big],\sqe \tag{\footnotesize \bf IHT} 
\end{align*}
where $\bs x^{(0)} = \bs 0$, and $\mu>0$ must satisfy $\mu^{-2} >
\|\bs\Phi\|:=\sup_{\bs u:\|\bs u\|=1} \|\bs\Phi\bs u\|$ for guaranteeing convergence~\cite{blumensath2011accelerated}. 

In other words, at each iteration, starting from the previous
estimation $\bs x^{(n)}$, the fidelity
function $\cl E(\bs u) := \tinv{2}\|\bs y - \bs \Phi\bs u\|^2$ is decreased by a
gradient descent step with gradient $\bs\nabla \cl E(\bs
x^{(n)}) = \bs\Phi^*(\bs\Phi\bs
x^{(n)}-\bs y)$, followed by a ``projection'' on $\Sigma_K$
accomplished by the hard thresholding $\cl H_K$.

In~\cite{blumensath2009iterative}, it is shown that if $\bs \Phi$
respects the \emph{restricted isometry property} (RIP) of order $3K$
with radius $\delta_{3K}<1/15$, which means that for all $\bs u\in\Sigma_{3K}$,
$(1-\delta_{3K})\|\bs u\|^2\leq \|\bs\Phi\bs u\|^2\leq (1+\delta_{3K})\|\bs
u\|^2$, then, at iteration $n^*= \lceil \log_2 \|\bs
x\|/\varepsilon\rceil$, the reconstruction error satisfies $\|\bs x - \bs x^{(n^*)}\| \leq 5 \varepsilon$.  

\sq
\section{Quantized Sensing Model}
\label{sec:framework}
\sq

For the sake of simplicity, let us consider a unit $K$-sparse signal $\xin \in \Sigma^*_K$ observed through the following Quantized
Compressed Sensing (QCS) model\sqe
\begin{equation}
  \label{eq:qcs-model}
  \bs y = \cl Q_b[\bs\Phi \xin],\sqe
\end{equation}
where $\bs \Phi\in\Rbb^{M\times N}$ is the sensing matrix and $\cl
Q_b$ the quantization operator defined at a \emph{resolution} of
$b$-bits per measurement, \ie with no further encoding treatment, $\bs
y$ requires a total of $\fk B = b M$ bits. In this work, we will not
consider any prequantization noise in \eqref{eq:qcs-model}. 

The quantization $\cl Q_b$ is assumed optimal with respect to the
distribution of each component of $\bs z = \bs \Phi \xin \in
\Rbb^M$. In particular, by considering only random Gaussian matrices
$\bs \Phi \sim \cl N^{M\times N}(0,1)$, \ie where each matrix entry
follows $\Phi_{ij} \sim_{\rm iid} \cl N(0,1)$, we
have $z_i \sim \cl N(0, \|\xin\|^2 = 1)$ and we
adjust $\cl Q_b$ to an optimal $b$-bits Gaussian Quantizer minimizing
the quantization distortion, \eg using
a Lloyd-Max optimization~\cite{gray1998quantization}. This provides a
set of thresholds $\{\tau_i \in \bar\Rbb: 1\leq i\leq 2^b+1\}$
(with $-\tau_1=\tau_{2^b+1}=+\infty$) defining $2^b$ quantization bins $\cl R_i=[\tau_i,\tau_{i+1})$, and a set of
levels $\{\qpt_i \in \cl R_i:1\leq i \leq 2^b\}$ such that\sqe 
$$
\cl Q_b[\lambda] = \qpt_k\ \Leftrightarrow\ \lambda \in \cl R_k,\sqe
$$
with $2\tau_i = \qpt_{i-1}+\qpt_i$ and $\qpt_i = \bb E[g_x| g_x\in
\cl R_i]$ with $g_x \sim \cl N(0, 1)$. Notice that this QCS model includes 1-bit CS scheme since $\cl
Q_1[\lambda] = \qpt_0\,\sign(\lambda)$ with
$\qpt_0 := \qpt_2=-\qpt_1 = \sqrt{2/\pi}$.

\sq
\section{Quantized Iterative Hard Thresholding}
\label{sec:quant-iter-hard}
\sq

In this section, we propose a generalization of the IHT algorithm
taking into account the particular nature of the scalar quantization
model introduced in Sec.~\ref{sec:framework}. The idea is to enforce
the consistency of the iterates with the quantized observations. This
is first achieved by defining an appropriate cost measuring
deviation from quantization consistency.  

Given $\nu,\lambda\in\bb R$ and using the levels and thresholds
associated to $\cl Q_b$, we first define\sqe
\begin{align}
  \label{eq:consist_energ_compos}
J(\nu, \lambda)&= \sum_{j=2}^{2^b} w_j\, \big|  \big( \sign(\lambda - \tau_j)\,
  (\nu - \tau_j) \big)_-\big|,\sqe\sqe    
\end{align}
with $w_j=\qpt_j-\qpt_{j-1}$. Equivalently, given $\cl
I(\nu,\lambda):=[\min(\nu,\lambda),\max(\nu,\lambda)]$, $J(\nu,
\lambda) = \sum_{j=2}^{2^b} w_j \chi_{\cl I}(\tau_j)\,|\nu - \tau_j|$.
The non-zero terms are therefore
determined by the thresholds lying between $\lambda$ and $\nu$, \ie
for which $\sign(\lambda - \tau_j) \neq \sign(\nu - \tau_j)$. Interestingly, $J(\nu; \lambda) =
J(\nu; \cl Q_b(\lambda))$ since $\sign(\lambda - \tau_j) =
\sign(\cl Q_b(\lambda) - \tau_j)$ for all $j\in [2^b+1]$. 

Then, our
quantization consistency function between two vectors $\bs u, \bs v \in \bb R^M$ reads\sqe
\begin{equation}
  \label{eq:eq:consist_energ}
  \cl J(\bs u, \bs v) := \sum_{k=1}^M J(u_k, v_k) = \cl J(\bs u, \cl
  Q_b(\bs v)).\sqe  
\end{equation}

This cost, which is convex with respect to $\bs u$, has two interesting limit cases. First, for $b=1$, it
reduces to the cost on which relies the \emph{binary
  iterative hard thresholding} algorithm (BIHT) adapted to 1-bit CS~\cite{jacques2011robust}. In this
context, the sum in \eqref{eq:consist_energ_compos} has only one term
(for $j=2$) and $\cl J(\bs u, \bs v) = 2\qpt_0 \,\|(\sign(\bs v)\odot\bs u
)_-\|_1$. Up to a normalization by $2\qpt_0$, this is the $\ell_1$-sided norm minimized by BIHT
which vanishes when $\qpt_0\,
\sign(\bs u) = \cl Q_1(\bs u) = \cl Q_1(\bs v) = \qpt_0\,
\sign(\bs v)$, with $\qpt_0$ defined in Sec.~\ref{sec:framework}.

\begin{figure}
  \centering
  \includegraphics[width=.6\columnwidth]{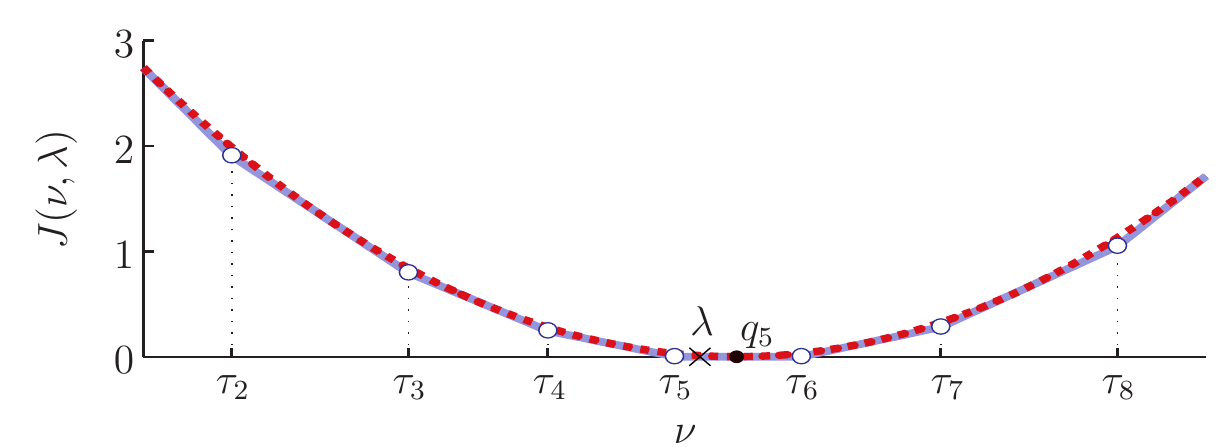}\\[-1mm]
  \caption{(plain curve) Plot of $J$ as a function of $\nu\in\Rbb$
    for $b=3$ ($\tau_5=0$) and $\lambda\in\cl R_5$. (dashed curve)
    Plot of $\tinv{2}(\nu-\qpt_5)^2$.}
  \label{fig:j-trend}
\end{figure}

Second, in the high resolution limit when $b\gg 1$, $\cl J(\bs u,\bs v)$ tends to
$\tinv{2}\|\bs u - \bs v\|^2$. Indeed, in this case $w_j \ll 1$ and, 
the sum in \eqref{eq:consist_energ_compos} tends to 
$$
J(\nu,\lambda) \simeq \big| {\textstyle\int}_{\nu}^{\lambda} (\nu -t)\, \ud t\, \big| = \tinv{2}
(\nu - \lambda)^2. 
$$
This asymptotic quadratic behavior of $J$ is illustrated in Fig.~\ref{fig:j-trend}.

Given the quantization consistency cost $\cl J$, we can now formulate a
generalization of \eqref{eq:lasso-iht} for estimating a $K$-sparse
signal $\xin$ observed by the model \eqref{eq:qcs-model}: 
\begin{equation}
  \label{eq:lasso-qiht}
\min_{\bs u\in\Rbb^N} \cl E_b(\bs u)\
\text{s.t.}\ \|\bs u\|_0\leq K,  
\end{equation}
with $\cl E_b(\bs u) := \cl J(\bs \Phi\bs u, \bs y) = \cl J(\bs
\Phi\bs u, \cl Q_b[\bs\Phi \xin])$. 

Following the procedure determining the IHT algorithm from \eqref{eq:lasso-iht}
(Sec.~\ref{sec:noisy-compr-sens}), our aim is to find an IHT
variant which minimizes the quantization inconsistency, as measured by
$\cl E_b$, instead of the quadratic cost $\cl E$. This is done by
first determining a \emph{subgradient} of the convex but non-smooth
function $\cl E_b$~\cite{rockconvex}.  

A quick calculation shows that a subdifferential of $J(\nu, \lambda)$ with respect to
$\nu$ reads\sqe 
\begin{equation}
  \label{eq:J-subgrad}
  \sum_{j=k_-+1}^{k_+} \tfrac{w_j}{2} ( \sign (\nu - \tau_j) - \sign
(\lambda - \tau_j) ),\sqe
\end{equation}
where $k_-=\min(k_\nu,k_\lambda)$, $k_+=\max(k_\nu,k_\lambda)$, and $k_\nu$ and $k_\lambda$ are the bin indices of $\cl Q_b(\nu)$ and
$\cl Q_b(\lambda)$ respectively. From the definition of the $w_j$, the
sum simplifies to $\qpt_{k_\nu} - \qpt_{k_\lambda}$. Therefore, a subgradient of $\cl J(\bs u, \bs v)$ with respect
to $\bs u$ reads simply $\cl Q_b(\bs u) - \cl Q_b(\bs v)$, so that a
subgradient of $\cl J(\bs \Phi\bs u, \bs y)$ with respect to $\bs
u$ corresponds to $\bs \Phi^*(\cl Q_b(\bs\Phi \bs u) - \bs y)$.

Therefore, from this last ingredient, we define the \emph{quantized
  iterative hard thresholding algorithm} (QIHT) by the recursion 
\begin{align*}
\hspace{-1.7mm}\bs x^{(n+1)} = \cl H_K\big[\bs x^{(n)} + \mu\bs\Phi^*\big (\bs y - \cl Q_b(\bs\Phi\bs
x^{(n)})\big)\big],\tag{\footnotesize \bf QIHT}  
\end{align*}
where $\bs x^{(0)} = 0$ and $\mu$ is
set hereafter. 

\sq
\section{QIHT analysis}
\label{sec:qiht-analysis}
\sq

Despite successful simulations of sparse
signal recovery from quantized measurements (see Sec.~\ref{sec:experiments}), we were not able to prove
the stability and the convergence of the QIHT algorithm yet. However,
there exist a certain number of promising properties suggesting the
existence of such a result. The first one comes from a limit case analysis. Except for the normalizing factor $\mu$, QIHT at 1-bit ($b=1$) reduces to BIHT
\cite{jacques2011robust}. Moreover, when $b\gg 1$, $\cl Q_b[\bs z]
\simeq \bs z$ for $\bs z\in\Rbb^M$ and we recover the IHT algorithm. 
These limit cases are consistent with the previous
observations made above on the asymptotic behaviors of $\cl J$ in
these two cases. 

Second, as for the modified Subspace Pursuit algorithm~\cite{Dai2009}, QIHT is designed for improving
the quantization consistency of the current iterate with the quantized
observations. For the moment, the importance of this improvement can
only be understood in 1-bit. Given $0<\delta <1$,
when $M = O(\delta^{-1} K \log N)$ and with high probability on the
drawing of a random Gaussian matrix $\bs \Phi \sim \cl N^{M\times N}(0,1)$, 
$\|\tfrac{\bs a}{\|\bs a\|} - \tfrac{\bs b}{\|\bs b\|}\| \leq \delta$
if $\cl Q_1(\bs \Phi\bs a) =  \cl
Q_1(\bs \Phi\bs b)$ for all $\bs a, \bs b \in
\Sigma_K$~\cite{jacques2011robust}.  Actually, it is shown in Appendix~\ref{sec:angul-dist-bounds} that if no more than $r$
components differ between $\cl Q_1(\bs \Phi\bs a)$ and $\cl Q_1(\bs
\Phi\bs b)$, then, with high probability on $\bs\Phi$, 
\begin{equation}
  \label{eq:almost-consistent-relation}
  \|\tfrac{\bs a}{\|\bs a\|} - \tfrac{\bs b}{\|\bs
  b\|}\| \leq (\tfrac{K+r}{K})\,\delta,
\end{equation}
for $M = O(\delta^{-1} K\log MN)$. We understand then the beneficial
impact of any increase of consistency between $\cl Q_1(\bs\Phi \bs x^{(n)})$
and $\bs y$ at each QIHT iteration.

Third, the adjustment of $\mu$, which is decisive for QIHT efficiency, leads also to some interesting
observations. Extensive simulations not presented
here pointed us that, for $\bs \Phi \sim \cl N^{M\times M}(0,1)$, $\mu
\propto 1/M$ seems to be a universal rule of efficiency at any bit
rate. Interestingly, this setting was already characterized for IHT
where $\mu \simeq 1/(1+\delta_{2K})$ if the sensing matrix respects the RIP
property with radius $\delta_{2K}$~\cite{blumensath2011accelerated}.
Since $\bs \Phi/\sqrt{M}$ is RIP for $\bs\Phi \sim \cl N^{M \times N}(0,1)$
as soon as $M=O(K \log N/K)$ this is equivalent to impose $\mu \simeq 1/M$. 

At the other extreme, the rule $\mu \propto 1/M$ is also consistent with the following 1-bit analysis. In
\cite{plan2012robust}, it is shown that the mapping $\bs u \to
\sign(\bs \Phi \bs u)$ respects an interesting
property that we arbitrary call \emph{sign product
  embedding}\footnote{In~\cite{plan2012robust},
  more general embeddings than this of $\Sigma_K$ are studied.} (SPE):
\begin{proposition}
\label{def:bpe-def-unif}
Given $0<\delta<1$, there exist two constants $c,C>0$ such that, if
$M \geq C\delta^{-6} K \log N/K$, then, with a probability higher than $1-8\exp(-c\delta^2 M)$, $\bs\Phi \sim
\cl N^{M\times N}(0,1)$ satisfies 
\begin{equation}
  \label{eq:bin-rpd-emb}
 \big|\mu^*\scp{\sign(\bs\Phi\bs u)}{\bs \Phi \bs v} -
 \scp{\bs u}{\bs v}\big|\leq \delta,\quad \forall \bs u, \bs v \in \Sigma^*_K,
\end{equation}
with $\mu^*= 1/(\qpt_0\,M)$. When $\bs u$ is fixed, the condition on $M$ is relaxed to $M
\geq C\delta^{-2} K \log N/K$.
\end{proposition} 

When $\bs \Phi$ respects \eqref{eq:bin-rpd-emb}, we simply write that
$\bs \Phi$ is SPE$(\Sigma^*_K,\delta)$. When $\bs u$ is fixed, we say that $\bs \Phi$ is locally
SPE$(\Sigma^*_K,\delta)$ on $\bs u$. This SPE property leads to an interesting phenomenon.
\begin{proposition}
\label{prop:1-bit-Thresholding-guarantee-unif}
Given $\bs x \in \Sigma^*_K$ and let $\bs\Phi \in \Rbb^{M\times N}$
be a matrix respecting the local SPE$(\Sigma^*_{2K}, \delta)$ on $\bs
x$ for some
$0<\delta<1$. Then, given $\bs y = \cl Q_1[\bs \Phi \bs x] =
\qpt_0\,\sign (\bs \Phi \bs x)$, the
vector $$\hat{\bs x} := \tfrac{1}{\qpt^2_0 M}\cl H_K(\bs \Phi^* \bs y),$$
satisfies
$\|\bs x - \hat{\bs x}\| \leq 2\delta$.
\end{proposition}

\begin{proof}
  Let us define $\cl T_0 = \supp\bs x$, $\cl T
  = \cl T_0\cup \supp \hat{\bs x}$, and $\bs a = \tfrac{1}{\qpt^2_0 M}\bs
  \Phi^* \bs y = \mu^* \bs \Phi^*\sign (\bs \Phi \bs x)$ with $\hat{\bs x} = \cl
  H_K(\bs a)$. Then $\hat{\bs x}$ is also the best $K$-term approximation $\bs a_{\cl T}=\bs
  \Phi^*_{\cl T}\bs y$, so that $\|\bs x - \hat{\bs x}\| \leq \|\bs x -
  \bs a_{\cl T}\|+\|\hat{\bs x} - \bs
    a_{\cl T}\| \leq 2 \|\bs x - \bs a_{\cl T}\|$. Therefore, since $\|\bs
    x - \bs a_{\cl T}\| = \sup_{\bs w \in \Sigma^*_{\cl T}} \scp{\bs w}{\bs
      x - \bs a_{\cl T}}$ and $\bs \Phi$ is SPE$(\Sigma^*_{2K}, \delta)$, 
$\|\bs x - \hat{\bs x}\| \leq 2\sup_{\bs w \in \Sigma^*_{\cl T}} \big(
    \scp{\bs w}{\bs x} - \mu^*\scp{\bs \Phi\bs w}{\sign(\bs \Phi \bs
      x)}\big) \leq 2\sup_{\bs w \in \Sigma^*_{\cl T}}\big(
    \scp{\bs w}{\bs x} - \scp{\bs w}{\bs x} + \delta \big) = 2\delta$,
using $\supp(\bs x -
\bs a_{\cl T})\subseteq \cl T$ with $\#\cl T \leq 2K$.
\end{proof}

This proposition shows that a single hard thresholding of $\tfrac{1}{\qpt^2_0 M}\bs\Phi^*
\bs y$ already provides a good estimation of $\bs x$. Actually, from
the condition on $M$ for reaching the local SPE, we deduce that $\|\bs x
- \hat{\bs x}\| = O(\sqrt{K/M})$. This is quite satisfactory for such a
simple $\bs x$ estimation and it suggests setting $\mu \propto 1/M$ in
QIHT for $b=1$ where $\hat{\bs x}$ is related to $\bs x^{(1)}$.

Noticeably, it has been recently observed in 
\cite{bahmani13} that $\hat{\bs
  x}' :=\hat{\bs
  x}/\|\hat{\bs x}\|$ is actually solution of 
$$
\argmax_{\bs u\in\Rbb^N}\ \scp{\bs y}{\bs\Phi\bs u}\quad
{\rm s.t.}\quad \|\bs u\|_0 \leq K,
$$
for which there exists the weaker
error bound $\|\bs x
- \hat{\bs x}'\|^2 = O(\sqrt{K/M})$ when $\bs x$ is fixed~\cite{plan2012robust}.

\sq
\section{Experiments}
\label{sec:experiments}
\sq

\begin{figure}
  \centering
  \newlength{\comsz}
  \newlength{\comsep}
  \setlength{\comsz}{0.33\textwidth}
  \setlength{\comsep}{4mm}
  \hspace{-4mm}\subfigure{\includegraphics[width=\comsz]{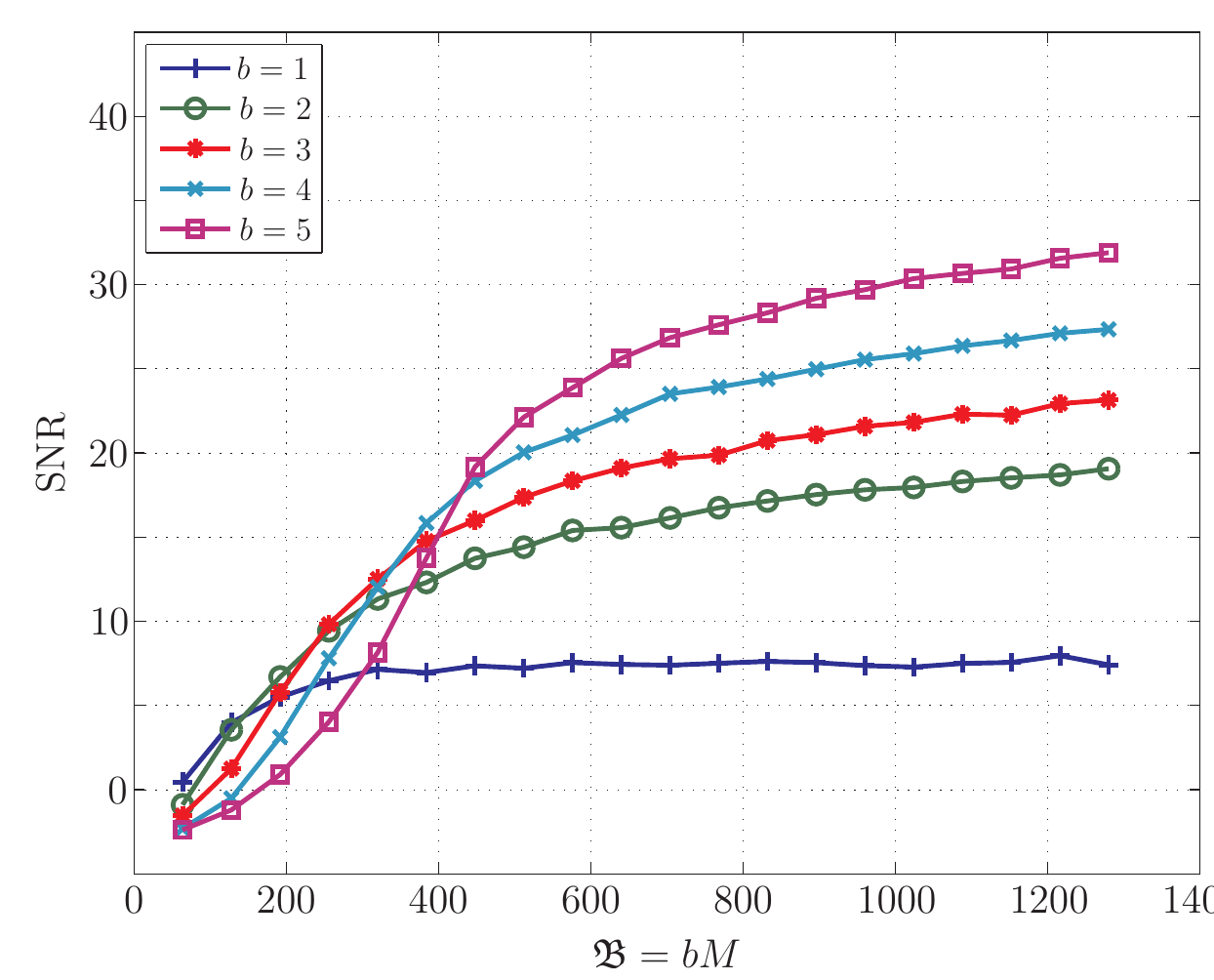}\hspace{-1.3cm}\raisebox{.65cm}{BPDN}}\hspace{\comsep}
  \subfigure{\includegraphics[width=\comsz]{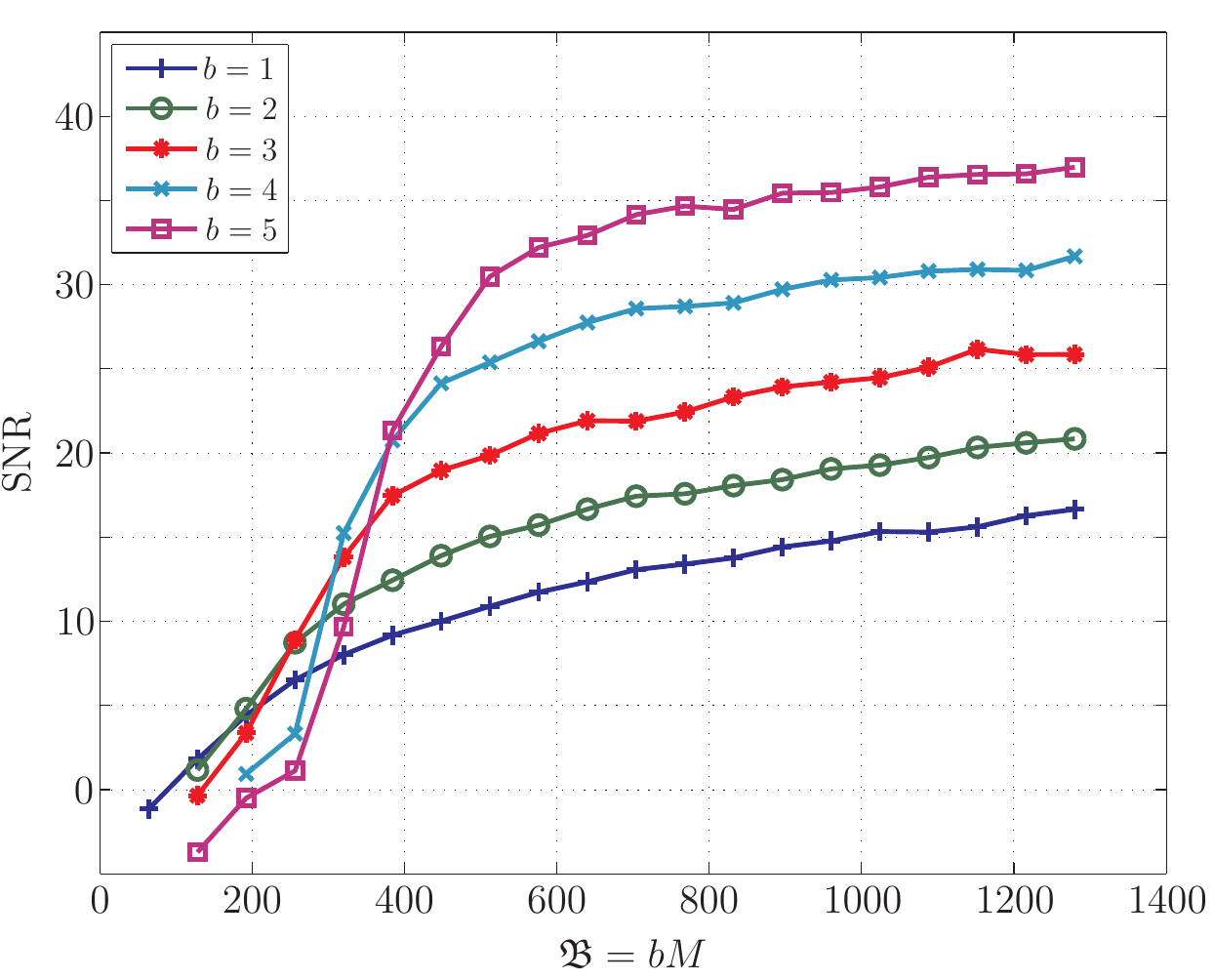}\hspace{-1.05cm}\raisebox{.65cm}{\
      IHT}}\hspace{\comsep}
  \subfigure{\includegraphics[width=\comsz]{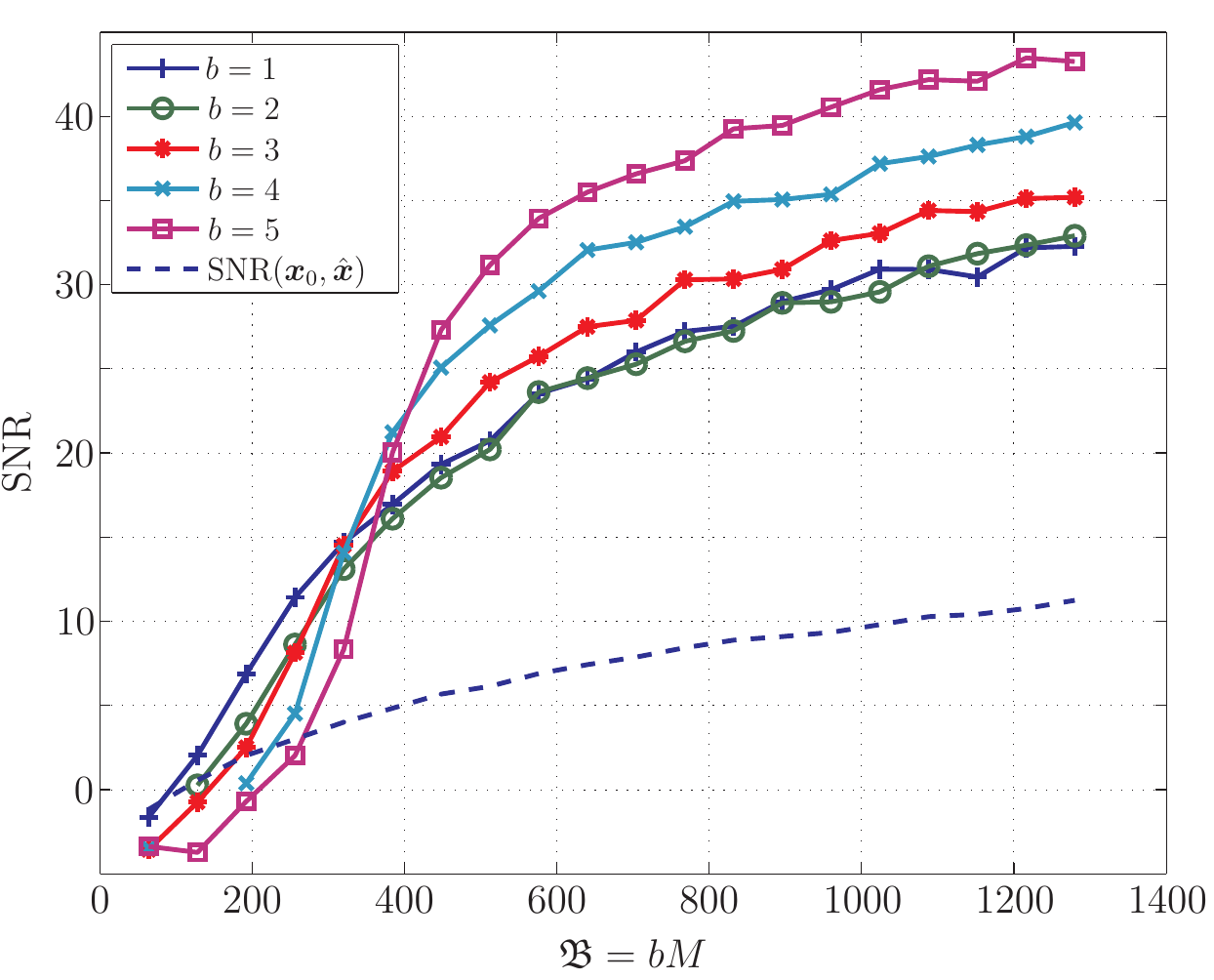}\hspace{-1.2cm}\raisebox{.65cm}{QIHT}}\hspace{\comsep}\,\\[-2mm]
  \caption{Comparison between (from left to right) BPDN, IHT and QIHT for several quantization scenarios. The SNR is expressed in dB as a function of the bit budget $\fk B$ and the number of bits $b$ used to quantize each measurement.}
  \label{fig:qiht_iht_bpdn_comp}
\end{figure}

An extensive set of simulations has been designed for evaluating the
efficiency of QIHT in comparison with two other methods more suited to
high-resolution quantization, namely, IHT and BPDN. Our objective is
to show that QIHT provides better quality results at least at small
quantization levels. For all experiments, we set $N=1024$, $K=16$ and the $K$-sparse signals were generated 
by choosing their supports uniformly at random amongst the $N \choose
K$ available ones, while their non-zero coefficients were drawn
uniformly at random on the sphere $S^{K-1}\subseteq \Rbb^K$.
For each algorithm, $100$ initial such sparse vectors were generated and the
reconstruction method
was tested for $1\leq b\leq 5$ and for $\fk B=bM \in \{64,
128,\,\cdots, 1280\}$, \ie approximately fixing $M = \lfloor \fk B/b \rfloor$.
For each experimental condition, the quantized $M$-dimensional 
measurement vectors $\bs{y}_b$ was generated as in
\eqref{eq:qcs-model} with a random sensing matrix $\bs \Phi \sim \cl
N^{M\times N}(0,1)$ and according to an optimal Lloyd-Max $b$-bits Quantizer
$\cl Q_b$ (Sec.~\ref{sec:framework}).
IHT and QIHT iterations were both stopped at step $n$ as soon as ${\| \bs x^{(n+1)}-\bs x^{(n)} \|}{\| \bs x^{(n+1)} \|}^{-1} <
10^{-4}$ or if $n=1000$. The BPDN algorithm was solved with the SPGL1 \textsc{Matlab} toolbox~\cite{spgl1}.
In IHT and QIHT, signal sparsity $K$ was assumed known 
and both were set with $\mu = \tinv{M}\big (1 -
\sqrt{\frac{2K}{M}}\big)$. This fits the IHT condition $\mu <
1/(1+\delta_{2K})$ mentioned in Sec.~\ref{sec:qiht-analysis} by
assuming that the RIP radius $\delta_{2K}$ behaves like $\sqrt{2K/M}$,
which is a common assumption in CS. For BPDN, the noise energy was given 
by an oracle installing BPDN in the best reconstruction scenario, \ie
$\epsilon = \|\bs \Phi \xin -\bs y\|_2$. Whatever the reconstruction
method, given an initial signal $\xin\in\Sigma^*_K$ and its
reconstruction $\bs x^*$, the reconstruction quality was measured by ${\rm SNR}(\xin,\bs x^*) =
- 20 \log_{10} \big\|\xin\, -\, \|\bs x^*\|^{-1} \bs x^*\big\|$. In
other words, we focus here on a good ``angular'' estimation of the
signals, adopting therefore a common metric for $b>1$ and for $b=1$, where
amplitude information is lost. Finally, for each method and each couple of $(M,b)$, the SNR was
averaged over the 100 test signals and expressed in dB. 

Fig.~\ref{fig:qiht_iht_bpdn_comp} gathers the SNR performances of the 3
methods as a function of $\fk B$. QIHT outperforms both
BPDN and IHT for the selected scenarios, especially for low bit
quantizers. At high resolution, the gain between QIHT and IHT
decreases as expected from the limit case analysis of QIHT. We can
also notice that, first, there is almost no quality difference
between QIHT at $b=1$ and $b=2$. This could be due to a non-optimality
of the Lloyd-Max quantizer with respect to QIHT reconstruction error minimization.
Second, BPDN and IHT asymptotically present the ``$6$dB per bit''
gain, while QIHT hardly exhibits such behavior only when $b=4\to 5$.   

Finally, in order to test Prop.~\ref{prop:1-bit-Thresholding-guarantee-unif}, 
the SNR reached by the single thresholding solution $\hat{\bs
  x}$ is plotted in dashed in Fig~\ref{fig:qiht_iht_bpdn_comp}-right. Despite its poor behavior compared to QIHT at $b=1$, it outperforms
BPDN at high $\fk B=M$ with a ${\rm SNR} \geq 10$dB at $M=N=1024$. A
curve fitting (no shown here) shows that this SNR increases
a bit faster than $20\log_{10} \sqrt{K/M} + O(1)$.

\sq
\section{Conclusion}
\label{sec:conclusion}
\sq

We have introduced the QIHT algorithm as a generalization of the BIHT
and IHT algorithms aiming at enforcing consistency with quantized observations at any bit resolution. In
particular, we showed that the almost obvious inclusion of the quantization operator
in the IHT recursion is actually related to the implicit minimization of a particular inconsistency
cost $\cl E_b$. This function generalizes
the one-sided $\ell_1$ cost of BIHT and asymptotically converges to
the quadratic fidelity minimized by IHT. There is still a hard work to
be performed in order to prove QIHT convergence and
stability. However, the different ingredients defining it, as $\cl
E_b$, deserve independent analysis extending previous $1$-bit embeddings developed in~\cite{jacques2011robust,plan2012robust,Plan2011}. 

\appendix

\section{Proximity of almost 1-bit consistent sparse vectors}
\label{sec:angul-dist-bounds}

The relation \eqref{eq:almost-consistent-relation} is induced by the
following theorem and by its subsequent
Corollary~\ref{cor:angul-bounds-almost}.
These use the normalized \emph{Hamming distance} between two strings $\bs
a,\bs b\in \{-1,+1\}^M$ defined by $d_{H}(\bs a, \bs b) = \frac{1}{M}\sum_{i=1}^{M} a_{i}
\oplus b_{i}$, where $\oplus$ is the XOR operation such that $a_i \oplus b_i$ equals 0 if $a_i=b_i$ and 1 otherwise.
For shortening the notations, we define also $\bs \varphi(\bs u) :=
\sign(\bs \Phi \bs u)\in\{-1,+1\}^M$ for $\bs u \in \Rbb^N$.

\begin{theorem}
\label{thm:1-bit-stable-embed}
Let $\bs \Phi\sim\cl N^{M\times N}(0,1)$. Fix $r\leq M/2$, $0\leq \eta\leq 1$ and $0<\delta<1$.  If the
number of measurements $M$ satisfies
\begin{equation}
\label{eq:numM}
M - r \geq \tfrac{2}{\delta}\,\big(2K\,\log(N) + r\log(M) + 4K
\log(\tfrac{17}{\delta}) + \log
\tfrac{2e}{\eta}\big),
\end{equation}
then, 
$$
\forall \bs a,\,\bs b \in\Sigma^*_K,\quad d_H\big( \bs \varphi (\bs a), \bs \varphi (\bs b)\big) \leq \tfrac{r}{M} \quad
\Rightarrow\quad \|\bs a -\bs b\|\leq \delta
$$
with probability exceeding $1-\eta$.
\end{theorem}
This improves the previous theorem proved in \cite{jacques2011robust}.

\begin{proof}
First, notice that if $M\,d_H\big( \bs \varphi (\bs a), \bs \varphi
(\bs b)\big) \leq r$, there exists a $\cl T\subset [M]$ with $|\cl
T|\geq M - r$ such that $\bs \varphi_{\cl T}(\bs a)=  \bs \varphi_{\cl
  T}(\bs b)$.  Let $[M]_r$ be the set of subsets of $[M]$ whose size is bigger than
$M-r$. Using a union bound argument, we have
\begin{align}
&\mathbb P\big[\ \exists\, \cl T\subset [M]_r,\ \exists\, \bs a, \bs
b\in\Sigma^*_K:\ \bs \varphi_{\cl T}(\bs a) =  \bs \varphi_{\cl T}(\bs
b),\ \|\bs a-\bs b\|> \delta\ \big]\nonumber
\\  
&\leq \bigcup_{\cl T\subset [M]_r} \mathbb P\big[\ \exists\, \bs a, \bs
b\in\Sigma^*_K:\ \bs \varphi_{\cl T}(\bs a) =  \bs \varphi_{\cl T}(\bs
b),\ \|\bs a-\bs b\|> \delta\ \big].
\label{eq:un-bound-anysubset}
\end{align}
We know from \cite[Theorem 2]{jacques2011robust} that, as soon as
$$
M' \geq \tfrac{2}{\delta}\,\big(2K\,\log(N) + 4K
\log(\tfrac{17}{\delta}) + \log
\tinv{\eta}\big),
$$  
the random generation of
$\bs \Phi' \sim \cl N^{M'\times N}(0,1)$ fullfils   
$$
\mathbb P\big[\ \exists \bs a, \bs b\in\Sigma^*_K:\ \bs \varphi'(\bs
a)=\bs \varphi'(\bs b),\ \|\bs a-\bs b\|> \delta\ \big]\ \leq\ \eta,
$$
with $\bs \varphi'(\cdot)=\sign(\bs \Phi' \cdot)$. Therefore, for a
given $\cl T \subset [M]_r$ and by setting $\bs \Phi'=(\Id_{\cl T})^T
\bs \Phi$, \ie the matrix obtained by restricting $\bs \Phi$ to the rows
indexed in $\cl T$, we have $\bs \varphi_{\cl T}=\bs \varphi'$, $M'=|\cl T|\geq M - r$ and  
$$
\mathbb P\big[\ \exists\, \bs a, \bs
b\in\Sigma^*_K:\ \bs \varphi_{\cl T}(\bs a) =  \bs \varphi_{\cl T}(\bs
b),\ \|\bs a-\bs b\|> \delta\ \big] \leq \eta
$$
if $M \geq r + \tfrac{2}{\delta}\,\big(2K\,\log(N) + 4K
\log(\tfrac{17}{\delta}) + \log
\tinv{\eta}\big)$.

Under the same condition on $M$, and observing that, for $r\leq \lfloor M/2\rfloor$, $|[M]_r| = \sum_{k=0}^r {M \choose M-k} \leq (r+1) {M \choose
  r}\leq (r+1) (e M/r)^r$, \eqref{eq:un-bound-anysubset} provides
$$
\mathbb P\big[\ \exists\, \cl T\subset [M]_r,\ \exists\, \bs a, \bs
b\in\Sigma^*_K:\ \bs \varphi_{\cl T}(\bs a) =  \bs \varphi_{\cl T}(\bs
b),\ \|\bs a-\bs b\|> \delta\ \big] \leq (r+1) (\tfrac{e M}{r})^r\,\eta.
$$

Analyzing the complementary event and redefining $\eta \leftarrow
(r+1) ({e M}/{r})^r\,\eta$, we get finally
\begin{align*}
&\mathbb P\big[\forall\, \bs a, \bs
b\in\Sigma^*_K:\ M\,d_H\big(\varphi_{\cl T}(\bs a),\bs \varphi_{\cl T}(\bs
b)\big) \leq r,\ \|\bs a-\bs b\|\leq \delta\ \big]\\
&=\ \mathbb P\big[\ \forall\, \cl T\subset [M]_r,\ \forall\, \bs a, \bs
b\in\Sigma^*_K:\ \bs \varphi_{\cl T}(\bs a) =  \bs \varphi_{\cl T}(\bs
b),\ \|\bs a-\bs b\|\leq \delta\ \big] \geq 1 - \eta,
\end{align*}
as soon as
$$
M \geq r + \tfrac{2}{\delta}\,\big(2K\,\log(N) + r\log(M) + 4K
\log(\tfrac{17}{\delta}) + \log
\tfrac{2e}{\eta}\big).
$$

\end{proof}

\begin{corollary}
\label{cor:angul-bounds-almost}
Let $\bs \Phi\sim\cl N^{M\times N}(0,1)$. Fix $r\leq M/2$, $0\leq \eta\leq 1$ and $0<\delta<1$.  If the
number of measurements $M$ satisfies
\begin{equation*}
M \geq \tfrac{2}{\delta}\,\big(2K\log(\max(N,M)) + 4K\log(\tfrac{17}{\delta}) + \log
\tfrac{2e}{\eta}\big),
\end{equation*}
then, 
$$
\forall \bs a,\,\bs b \in\Sigma^*_K,\quad d_H( A(\bs a), A(\bs b)) \leq \tfrac{r}{M} \quad
\Rightarrow\quad \|\bs a -\bs b\|\leq\ \tfrac {K + r}{K}\,\delta
$$
with probability exceeding $1-\eta$.
\end{corollary}
\begin{proof}
The proof is obtained from Theorem~\ref{thm:1-bit-stable-embed} by
redefining $\delta \leftarrow \tfrac{K+r}{K}\,\delta$, observing that
$\frac{2}{\delta}\log M\geq 1$ if $M > 1$ and by slightly enforcing
the condition \eqref{eq:numM} on $M$.  
\end{proof}

\end{document}